\documentclass[default]{sn-jnl}% Default
\usepackage{graphicx}%
\usepackage{multirow}%
\usepackage{amsmath,amssymb,amsfonts}%
\usepackage{amsthm}%
\usepackage{mathrsfs}%
\usepackage{xcolor}%
\usepackage{textcomp}%
\usepackage{manyfoot}%
\usepackage{booktabs}%
\usepackage{algorithm}%
\usepackage{algorithmicx}%
\usepackage{algpseudocode}%
\usepackage{listings}%

\usepackage{braket}
\usepackage{orcidlink}
\newtheorem{theorem}{Theorem}%  meant for continuous numbers
\newtheorem{definition}{Definition}%

\begin{document}

\title[Quantum search by continuous-time quantum walk on $t$-designs]{Quantum search by continuous-time quantum walk on $t$-designs}

%% OR: Quantum search in $t$-designs by continuous-time quantum walk


\author*[1,2]{\fnm{Pedro H. G.} \sur{Lugão} \orcidlink{0000-0002-2316-445X}}\email{pedrogasparetto@ymail.com}\equalcont{These authors contributed equally to this work.}

\author[1]{\fnm{Renato} \sur{Portugal}\orcidlink{0000-0003-0894-4279}}\email{portugal@lncc.br}
\equalcont{These authors contributed equally to this work.}

\affil*[1]{\orgname{National Laboratory of Scientific Computing - LNCC}, \orgaddress{\street{Av. Getúlio Vargas}, \city{Petrópolis}, \postcode{25651-075}, \state{RJ}, \country{Brazil}}}

\affil[2]{\orgname{Universidade Federal de Juiz de Fora - UFJF}, \orgaddress{\street{Rua José Lourenço Kelmer}, \city{Juiz de Fora}, \postcode{36036-900}, \state{MG}, \country{Brazil}}}


\abstract{This work examines the time complexity of quantum search algorithms on combinatorial $t$-designs with multiple marked elements using the continuous-time quantum walk. Through a detailed exploration of $t$-designs and their incidence matrices, we identify a subset of bipartite graphs that are conducive to success compared to random-walk-based search algorithms. These graphs have adjacency matrices with eigenvalues and eigenvectors that can be determined algebraically and are also suitable for analysis in the multiple-marked vertex scenario. We show that the continuous-time quantum walk on certain symmetric $t$-designs achieves an optimal running time of $O(\sqrt{n})$, where $n$ is the number of points and blocks, even when accounting for an arbitrary number of marked elements. Upon examining two primary configurations of marked elements distributions, we observe that the success probability is consistently $o(1)$, but it approaches 1 asymptotically in certain scenarios.}




\maketitle

\section{Introduction}

Continuous-time quantum walk, a quantum analog of the continuous-time Markov chain, is a dynamic paradigm that is crucial to quantum computation and information processing~\cite{FG98}. They are powerful computational procedures that encompass the non-classical properties of quantum systems, enabling them to explore superposition and entanglement with a straightforward formalism. Continuous-time quantum walk and the alternative versions in discrete-time are invaluable tools in constructing quantum search algorithms, where they can perform searches over graphs faster than their classical counterparts~\cite{CG04}. This speedup underlies their advantage in finding solutions to various optimization problems~\cite{MW20}. Furthermore, quantum walks offer a robust mechanism for simulating complex quantum physical systems, enabling the understanding and prediction of the behavior of quantum particles in various environments~\cite{KGK21}.

The fundamental concept behind a quantum walk-based search algorithm involves initializing a quantum walker at a state that is easily prepared, and then allowing it to evolve according to a standard quantum walk operator, which is modified by an oracle that identifies the locations of the marked vertices. When the quantum walk operator is aptly designed, the quantum walker will not only reach the marked vertex eventually but also display a significant probability of being found at one of the marked vertices. Indeed, the efficiency of a quantum walk-based search algorithm is determined by both the runtime and the success probability~\cite{Por18book}.

Given the current progress in quantum computing, it becomes compelling to explore quantum search algorithms on classes of graphs, as opposed to individual graphs, and to determine the time complexity of these algorithms relative to the number of vertices. Historically, the class of complete graphs was the first to undergo analysis. In this context, quantum-walk-based search mirrors the function of Grover's algorithm~\cite{Gro97}. Quantum search by continuous-time quantum walk on certain specific graph classes is analyzed in~\cite{CG04,PTB16,TSP22}.
The fascinating interplay between quantum computation and combinatorial graph theory persistently unveils geometrical structures that hold potential for devising new quantum algorithms.

One structure in combinatorial mathematics is the block design~\cite{Sti04,BJL99}. It refers to an incidence structure comprising a set of points and a selection of subset families, termed blocks. The points within these blocks are chosen to meet specific frequency conditions. This meticulous selection ensures that the entire collection of blocks exhibits a sense of symmetry, often termed ``balance''.
In the absence of additional context, the term ``block design'' is typically interpreted as a balanced incomplete block design or, equivalently, a 2-design. Historically, this specific type has been extensively researched due to its pivotal role in experimental design, coding theory, cryptography, and software testing~\cite{Rot76}. Such applications may also be influential in the quantum realm, making it intriguing to examine the relationship between $2$-designs and quantum algorithms from a theoretical standpoint. A more generalized version of this concept is the $t$-design, which is the structure of our interest.

While our primary focus in this paper lies on combinatorial $t$-designs, it's worth noting the intriguing extension of these designs into the realm of quantum mechanics, known as quantum $t$-designs~\cite{DCEL09}.  In the context of quantum information theory, a set of quantum states forms a quantum $t$-design if the average of certain quantum operations, specifically those expressible as polynomials of degree $t$ or less in the entries of a density matrix, over this set mirrors that of the entire space of quantum states (according to the Haar measure). Essentially, quantum $t$-designs offer a compact and representative subset of quantum states, capturing the statistical properties of random quantum states up to the $t$-th moment. These designs have found applications in quantum algorithms~\cite{AE07} and quantum error correction~\cite{EAZ05}.

Turning our attention back to the concept of combinatorial $t$-designs, the incidence matrix of a $t$-design enumerates the repetitions of each point within every block. As the incidence matrix constitutes a bipartite graph~\cite{GR01}, a continuous-time quantum walk on a $t$-design is synonymous with a walk on its corresponding bipartite graph. Apart from the complete bipartite graph, there are limited results in existing literature that establish the time complexity of search algorithms on bipartite graphs. To our best understanding, none of these address scenarios involving multiple marked cases in the continuous-time model. Indeed, quantum-walk-based search algorithms have a rich history in the single-marked case, beginning with the foundational paper by Shenvi \textit{et al.}~\cite{SKW03} and Childs\&Goldstone~\cite{CG04} on searching for a single marked vertex on hypercubes and other graphs. However, the focus has recently shifted to the multiple-marked case~\cite{Won16,BLP21,LPST23,ACNR22,RKM23}.

The primary goal of this work is to determine the time complexity of quantum search algorithms on $t$-designs with multiple marked elements (points or blocks) using the continuous-time quantum walk. Such effort helps in discerning the time complexity of quantum search on bipartite graphs, a notably challenging problem. Our focus lies in pinpointing specific subsets of bipartite graphs that are most conducive to producing successful results. Ideally, these graphs would have adjacency matrices with determinable eigenvalues and eigenvectors in algebraic terms, underlining their suitability for evaluating the time complexity in the multiple-marked vertex context. We use $t$-designs to identify such a significant subset of bipartite graphs. Our findings indicate that the optimal running time of a continuous-time quantum walk on certain symmetric $t$-designs is $O(\sqrt{n})$, where $n$ stands for the number of points and blocks, even when accounting for an arbitrary number of marked elements. We assess two configurations: in the first, all marked elements are situated in one part of the bipartite graph, meaning that all of them are points or all of them are blocks; in the second, the marked elements are evenly distributed between the two parts. The probability of success remains consistently $o(1)$ across both setups, but it approaches 1 in particular scenarios.

This paper is organized as follows. In Section~\ref{sec:spec}, we derive the spectral decomposition of symmetric $t$-designs. In Section~\ref{sec:search}, we first determine the time complexity of a quantum search using a continuous-time quantum walk on symmetric $t$-designs, beginning with the single-marked case. We then address the two-marked case, and subsequently, build up to the multiple-marked scenario. We conclude our discussion in Section~\ref{sec:conc}, where we present our final thoughts and conclusions.

\section{Spectral decomposition of symmetric \textit{t}-designs}\label{sec:spec}

Let $X$ be a finite set of $v$ elements called points. Let $k$ and $\lambda$ be positive integers. 
The definition of a $t$-design, which includes the balanced incomplete block design (2-design) is as follows.

\begin{definition}
[\textbf{$t$-design}]
Given a set $X$ of $v<\infty$ elements (called points) and any positive integer $t$, assuming that $t\le k\le v$, a $t$-$(v,k,\lambda)$-design is a class of $k$-subsets of $X$, called blocks, such that every point $x$ in $X$ appears in exactly $r$ blocks, and every $t$-subset appears in exactly $\lambda$ blocks. 
\end{definition}

Let us define
$$\lambda_{i}=\lambda {\binom {v-i}{t-i}}{\binom {k-i}{t-i}}^{-1}$$
for $i=0,1,\ldots ,t$. Note that  $\lambda_t=\lambda$ and $\lambda_i$ represents the number of blocks that contain any $i$-set of points.

The number of blocks $b$ in a $t$-design is given by
\[ 
b=\lambda_0=\lambda {v \choose t}{k \choose t}^{-1}.
\]
Number $r$ is given by
\[
r=\lambda_1=\lambda {v-1 \choose t-1}{k-1 \choose t-1}^{-1}.
\]
Any $t$-$(v,k,\lambda)$-design is also an $s$-$(v,k,\lambda)$-design for any $s$ satisfying $1 \le s \le t$~\cite{Sti04}. It should be observed that the $\lambda$ value changes as described above and is dependent on $s$. It follows that every $t$-design with $t \ge 2$ is also a 2-design. Besides, the number of blocks $\lambda_i$ that contain any $i$-subset of $X$ in a $t$-design is independent of the choice of the subset for $i=1,\ldots,t$.  The term block design by itself usually means a 2-design. 
Given numbers $t$, $v$, $k$, and $\lambda$, it is no easy task to find examples of $t$-$(v,k,\lambda)$-designs.

Let $\mathcal{D}$ be a $t$-design with parameters $(v,b,r,k,\lambda_t)$ and $G$ its incidence graph. By definition, $G$ is a $(r,k)$-biregular bipartite graph~\cite{GR01}. The adjacency matrix of $G$ is
\begin{equation}
    A = \begin{pmatrix}
        0 & N\\
        N^T & 0
    \end{pmatrix},
\end{equation}
where $N$ is a $v\times b$ incidence matrix.

\begin{theorem}
Let $\mathcal{D}$ be a symmetric $t$-design (i.e., $v=b$ and $r=k$). Then, the eigenvalues of the adjacency matrix of $\mathcal{D}$ are $\{\pm k,\pm \sqrt{r-\lambda_2}\}$ and the spectral idempotents are
\renewcommand*{\arraystretch}{1.5}
\begin{align}
    E_{-k} &= \frac{1}{2v}\begin{pmatrix}
        J & -J\\
        -J & J
    \end{pmatrix},\\
    E_{k} &= \frac{1}{2v}\begin{pmatrix}
        J & J\\
        J & J
    \end{pmatrix},\\
    E_{-\sqrt{r-\lambda_2}} &= \frac{1}{2}\begin{pmatrix}
        \mathrm{I} - \frac{1}{v}J & -M\\
        -M^T & \mathrm{I} - \frac{1}{v}J
    \end{pmatrix},\\
    E_{\sqrt{r-\lambda_2}} &= \frac{1}{2}\begin{pmatrix}
        \mathrm{I} - \frac{1}{v}J & M\\
        M^T & \mathrm{I} - \frac{1}{v}J
    \end{pmatrix},
\end{align}
where $M  = \frac{1}{v\sqrt{r-\lambda_2}} \left(-kJ + vN\right)$ and $J$ is a matrix of ones.
\end{theorem}  %%% M = -\frac{k}{v\sqrt{r-\lambda_2}} J + \frac{1}{\sqrt{r-\lambda_2}}N

\begin{proof}  To proof this theorem, it suffices to show that \{$E_k,E_{-k},E_{\sqrt{r-\lambda_2}} ,E_{-\sqrt{r-\lambda_2}}$\} is a set of orthogonal projectors obeying the completeness relation $\sum_i E_i = \mathrm{I}$ and that the spectral decomposition of $A$ is
    \begin{equation}\label{eq:decomp}
        A = k E_k + (-k) E_{-k} + \sqrt{r-\lambda_2} E_{\sqrt{r-\lambda_2}} + (-\sqrt{r-\lambda_2}) E_{-\sqrt{r-\lambda_2}}.
    \end{equation}
The following equations are valid identities:
     \begin{enumerate}
         \item  $J^2 = vJ$ \label{id1}
         \item  $(I - J/v)^2 = (I-J/v)$\label{id2}
         \item $JN=NJ = N^TJ = JN^T = kJ$\label{id3}
         \item $JM = MJ = M^TJ = JM^T = 0$\label{id4}
         \item $N^TN = NN^T = (r-\lambda_2)(I-J/v) + k^2J/v$\label{id5}
         \item $MM^T = M^TM = I-J/v$\label{id6}
         \item $J(I-J/v) = 0$\label{id7}
     \end{enumerate}
Using Identities~\eqref{id1}-\eqref{id6}, we show that $E_{\pm k}^2=E_{\pm k}$ and $E_{\pm\sqrt{r-\lambda_2}}^2=E_{\pm\sqrt{r-\lambda_2}}$.
Using additionally Identity~\eqref{id7}, we show that these projectors are pairwise orthogonal.
Summing the matrices we can see that $\sum_i E_i = \mathrm{I}$.
Finally using the definition of $M$ we check that Equation~\eqref{eq:decomp} is valid, which concludes the demonstration.
\end{proof}

\section{Quantum search}\label{sec:search}

Quantum search by continuous-time quantum walk on a graph with adjacency matrix $A$ is driven by a Hamiltonian $H$, whose expression is
\begin{equation*}
H = -\gamma A - \sum_{w\in W}\ket{w}\bra{w},
\end{equation*}
where $\gamma$ is a positive parameter and $W$ represents the set of marked vertices~\cite{CG04}.  The evolution operator is given by
$$U(t)=\text{e}^{-\text{i}Ht},$$
and the state of the walk at time $t$ is $\ket{\psi(t)}=U(t)\ket{\psi(0)}$, where $\ket{\psi(0)}$ denotes the initial state. The probability of finding the walker on vertex $w$ at time $t$ is $p_w(t)=\left|\langle w|\psi(t)\rangle\right|^2$. For our initial state, we assume an unbiased uniform superposition. Our goals include (1) determining the optimal running time $t_\text{opt}$ of the search algorithm as a function of the parameters of the $t$-design, and (2)~determining its success probability, which is $\sum_{w\in W} p_w(t_\text{opt})$.

Now, let's assess the time complexity of a quantum search using a continuous-time quantum walk on symmetric $t$-designs. We will begin our analysis with $t$-designs that contain a single marked vertex. Specifically, $W=\{0\}$, where 0 stands for the label of a point or a block.

\subsection{Single marked vertex}
Assuming that the $t$-design is symmetric and contains only one marked vertex, labeled 0, we next consider $\{\phi_l\}_{l=0}^{3}$ as the list of eigenvalues of the $t$-design, ordered from highest to lowest.

Using the method outlined in~\cite{CNR20b,CGTX22,SPP23,LPST23}, we have to calculate the  sums 
\begin{align*}
    S_1 &= \sum_{l=1}^3\frac{\|E_{\phi_l}\ket{0}\|^2}{\phi_0 - \phi_l},\\
    S_2 &=  \sum_{l=1}^3\frac{\|E_{\phi_l}\ket{0}\|^2}{(\phi_0 - \phi_l)^2},
\end{align*}
which are equal to
\begin{equation*}
    S_1 =\frac{v-1}{v}\frac{k}{k^2 - r + \lambda_2} + \frac{1}{4vk}
\end{equation*}
and
\begin{equation*}
    S_2 =\frac{v-1}{v}\frac{(k^2 + r - \lambda_2)}{(k^2 - r + \lambda_2)^2} + \frac{1}{8vk^2}.
\end{equation*}

Now, take into consideration that the evolution operator is $\text{e}^{-iHt}$, where $H = -\gamma A - \ket{0}\bra{0}$, the optimal value for $\gamma$ along with its corresponding optimal hitting time are provided by
\begin{align*}
    \gamma &= S_1,\\
    t_{\text{opt}} &= \frac{\pi}{2\epsilon},
\end{align*}
where $\epsilon = \frac{S_1\|E_{k}\ket{0}\|}{\sqrt{S_2}}$. Using that $\lambda_2 = r(k-1)/(v-1)$ and $k=r$, we obtain
\begin{equation*}
    t_{\text{opt}} = \frac{\pi  \sqrt{k+1}}{\sqrt{2k}} \sqrt{v} + O\left(\frac{1}{\sqrt{v}}\right).
\end{equation*}
Given that the number of vertices, ${n}$, equals $v + b = 2v$, it follows that the optimal time $t_{\text{opt}}$ is on the order of $\sqrt{{n}}$. To complete the time complexity analysis, we must calculate the success probability, which is
\begin{equation*}
    p_\text{succ} = \frac{S_1^2}{2vS_2\|E_{k}\ket{0}\|^2} = \frac{k}{k+1} + O\left(\frac{1}{v}\right).
\end{equation*}

It's interesting to note that, if the bipartite graph is complete, then $k = v$. As a result, the success probability becomes $p_\text{succ} = 1 + O(1/v)$ and the optimal running time is $t_\text{opt}=\pi\sqrt v/\sqrt 2 + O(1/\sqrt v)$. These are the expected results for a complete bipartite graph as shown in~\cite{PM22}. 

While Theorem 1 doesn't apply to complete graphs due to an indeterminacy in $M$ when $k = v$ (since it implies $\lambda_2 = r$), it is indeed a compelling result that the final asymptotic behavior aligns with the theory as $k$ approaches $v$. As we increase the number of marked vertices, we may or may not observe this same correspondence.
\subsection{Multiple marked vertices}

Adhering to the methodology presented in~\cite{LPST23} for multiple marked vertices, where $W$ is the set of marked vertices, our interest remains focused on only two eigenvalues of the Hamiltonian, which are given by
\begin{equation}
    \lambda^\pm = -\gamma \phi_0 \pm \epsilon.\label{eq:lambdaEps}
\end{equation}

The extension of the one-marked case to the multiple-marked case for the computation of these two eigenvalues implies that we must solve the equation
\begin{equation}
\det(\Lambda_\lambda) = 0,\label{eq:det}
\end{equation}
where $\Lambda_\lambda$ is a $|W|\times|W|$ matrix defined as
\begin{equation}
    (\Lambda_\lambda)_{ww'} =  \delta_{ww'} + \sum_{l=0}^3 \frac{\bra{w}E_{\phi_l}\ket{w'}}{\lambda + \gamma \phi_l}.\label{eq:LambdaTerm}
\end{equation}

By combining Equations~\eqref{eq:LambdaTerm} and~\eqref{eq:lambdaEps}, we derive an expression for $\Lambda_{\lambda(\epsilon)}$, denoted as $\Lambda_\epsilon$, accurate up to the second order in $\epsilon$, which is 
\begin{equation}
    (\Lambda_\epsilon)_{ww'} = \pm \frac{\bra{w}E_k\ket{w'}}{\epsilon} + \delta_{ww'} - \frac{S^{(1)}_{ww'}}{\gamma} \mp \frac{\epsilon S^{(2)}_{ww'}}{\gamma^2},
\end{equation}
where
\begin{equation}
    S^{(1)}_{ww'} = \sum_{l=1}^{3} \frac{\bra{w}E_{\phi_l}\ket{w'}}{k-\phi_l}
\end{equation}
and
\begin{equation}
    S^{(2)}_{ww'} = \sum_{l=1}^{3} \frac{\bra{w}E_{\phi_l}\ket{w'}}{(k-\phi_l)^2}.
\end{equation}

Despite this modification, Equation~\eqref{eq:det} remains challenging to compute for a general case. Therefore, we confine our analysis to certain specific cases where the determinant can be more readily managed.

\subsubsection*{Two-marked vertices}

We now proceed to analyze the two-marked case to aid in the computation of the more challenging multiple-marked case. Since we only have two marked vertices, $\Lambda_\lambda$ is a $2\times 2$ matrix. Although the determinant is straightforward to compute, it's important to note that $(\Lambda_\lambda)_{ww'}$ with $w\neq w'$ can have different expressions depending on the location and relation between the marked vertices, as determined by the values of $S^{(1)}_{ww'}$ and $S^{(2)}_{ww'}$. Thus, even with just two marked vertices, we must consider three different cases to exhaust all possibilities using that the bipartite graph has two parts $V$ and $V'$.

\subsubsection*{Case 1: $w$ and $w'$ adjacent}

In the scenario that $w\in V$ and $w'\in V'$ belong to different parts but they are adjacent, we have
\begin{align*}
    S^{(1)}_{ww'}&=\frac{(-k+v)k}{v\sqrt{r-\lambda_2}(k^2 - r + \lambda_2)} - \frac{1}{4vk},\\
    S^{(2)}_{ww'}&=\frac{(-k+v)(k^2 + r - \lambda_2)}{v\sqrt{r-\lambda_2}(k^2 - r + \lambda_2)^2} - \frac{1}{8vk^2}.
\end{align*}
Given that $\Lambda_\lambda$ is a symmetric matrix, we find that
$$\det(\Lambda_\lambda) = \big((\Lambda_\lambda)_{ww} + (\Lambda_\lambda)_{ww'}\big)\big((\Lambda_\lambda)_{ww} - (\Lambda_\lambda)_{ww'}\big).$$
Term $(\Lambda_\lambda)_{ww} + (\Lambda_\lambda)_{ww'}$ is the eigenvalue of $\Lambda_\lambda$ associated with the uniform eigenvector. The correct values for $\lambda^\pm$ are obtained equating this term to zero. 
We select $\epsilon$ such that
\begin{align*}
    0 &= \epsilon \Big((\Lambda_\epsilon)_{ww} + (\Lambda_\epsilon)_{ww'}\Big)\\
    &= a(\gamma) + b(\gamma)\epsilon + c(\gamma)\epsilon^2 + O(\epsilon^3),
\end{align*}
with $\gamma$ chosen in a manner such that $b(\gamma) = 0$, leading us to two symmetrical roots. This approach yields
\begin{equation}
    \gamma=\frac{k\left(k - v + (1-v)\sqrt{k-\lambda_2}\right)}{v\sqrt{k-\lambda_2}\,(k-\lambda_2 - k^2)}
\end{equation}
and
\begin{equation}
    \epsilon^\pm = \pm \frac{\sqrt{k + \sqrt{k}}}{\sqrt{k+1}}\frac{1}{\sqrt{v}} + O\left(\frac{1}{v}\right).
\end{equation}
By applying the formulas for the two-marked case as outlined in~\cite{LPST23}, we obtain
\begin{equation*}
    t_\text{opt} = \frac{\pi\sqrt{k+1}}{2\sqrt{k+\sqrt{k}}}\sqrt{v} + O\left(\frac{1}{v}\right)
\end{equation*}
and
\begin{equation*}
   p_\text{succ} = \frac{4\sqrt{k}(\sqrt{k}+1)^3(k+1)}{(4\sqrt{k} + k + 2k\sqrt{k} + 1)^2} + O\left(\frac{1}{\sqrt{v}}\right).
\end{equation*}

\subsubsection*{Case 2: $w$ and $w'$ are non-adjacent and belong to different parts}

In this scenario, we have
\begin{align*}
    S^{(1)}_{ww'}&=\frac{-k^2}{v\sqrt{r-\lambda_2}(k^2 - r + \lambda_2)} - \frac{1}{4vk},\\
    S^{(2)}_{ww'}&=\frac{-k(k^2 + r - \lambda_2)}{v\sqrt{r-\lambda_2}(k^2 - r + \lambda_2)^2} - \frac{1}{8vk^2}.
\end{align*}
This minor difference in the scenario results in a different $\gamma$ value
\begin{equation}
    \gamma=\frac{k\left(k + (1-v)\sqrt{k-\lambda_2}\right)}{v\sqrt{k-\lambda_2}\,(k-\lambda_2 - k^2)}
\end{equation}
and in a simpler $\epsilon$
\begin{equation}
    \epsilon^\pm = \pm \frac{\sqrt{k}}{\sqrt{k+1}}\frac{1}{\sqrt{v}} + O\left(\frac{1}{v}\right).
\end{equation}
With those new expressions, we have
\begin{equation*}
    t_\text{opt} = \displaystyle \frac{\pi  \sqrt{k+1}}{2 \sqrt{k}}\sqrt{v} + O\left(\frac{1}{\sqrt{v}}\right)
\end{equation*}
and
\begin{equation*}
    p_\text{succ} = \frac{k}{k+1} - O\left(\frac{1}{\sqrt{v}}\right).
\end{equation*}

\subsubsection*{Case 3: $w$ and $w'$ are in the same part}

This is the simplest scenario ($\{w,w'\}\in V$ or $\{w,w'\}\in V'$) , since the sums are
\begin{align*}
    S^{(1)}_{ww'}&=-\frac{k}{v(k^2 - r + \lambda_2)} + \frac{1}{4vk},\\
    S^{(2)}_{ww'}&=-\frac{k^2+r-\lambda_2}{v(k^2 - r + \lambda_2)^2} + \frac{1}{8vk^2}.
\end{align*}
Then, we obtain
\begin{equation}
    \gamma =\displaystyle \frac{2 k^{2} v - 3 k^{2} + \lambda_{2} - r}{2 k v \left(k^{2} + \lambda_{2} - r\right)}
\end{equation}
and
\begin{equation}
    \epsilon^\pm = \pm \frac{\sqrt{k}}{\sqrt{k+1}}\frac{1}{\sqrt{v}} + O\left(\frac{1}{v}\right),
\end{equation}
which is the same first-order expression of Case 2. The expression for the optimal running time and success probability are
\begin{equation*}
    t_\text{opt} = \displaystyle \frac{\pi  \sqrt{k+1}}{2 \sqrt{k}}\sqrt{v} + O\left(\frac{1}{\sqrt{v}}\right)
\end{equation*}
and
\begin{equation*}
    p_\text{succ} = \frac{k}{k+1} - O\left(\frac{1}{\sqrt{v}}\right),
\end{equation*}
which are the same first-order expressions of Case 2.

\subsubsection*{All marked vertices are in the same part}

To generalize this analysis, we restrict our attention to the case where there are $m$ marked vertices in one of the parts of the bipartite graph. This is the only situation where all vertices are related in the same way (as described in the previous subsection), allowing us to fully compute the determinant of the $\Lambda_\lambda$ matrix.

The key concept here is that $(\Lambda_\lambda)_{ww'}$ is the same for every pair $(w,w')$ in the same part ($V$ or $V'$). This fact can be verified by noting that the sums considered in the previous section only depend on the relation between the two vertices, not their position. Consequently, because every pair shares the same relation, we obtain an $m\times m$ $\Lambda_\lambda$ matrix in the following form
\begin{equation*}
    \begin{pmatrix}
        a & b &  \hdots &b\\
        b & a &  \hdots &b\\
        \vdots & \vdots & \ddots & \vdots\\
        b & b & \hdots & a 
    \end{pmatrix},
\end{equation*}
whose determinant is $(a-b)^{m-1}(a+(m-1)b)$. This can be easily checked noting that $(1,1,...,1)$ is a $(a+(m-1)b)$-eigenvector and that $\dim(\ker(\Lambda_\lambda - I(a-b))) = m-1$, giving us an $(a-b)$-eigenspace of dimension $m-1$. 

Let's recall that $a = (\Lambda_\lambda)_{ww}$ and $b = (\Lambda_\lambda)_{ww'}$, with $w\neq w'$. Our goal is to find a root for the determinant, so we can either take $\lambda$ as a root of $(a-b)$ or a root of $(a+(m-1)b)$. Our calculations show that the $\lambda$ which serves as the root of $(a-b)$ does not follow the format of Equation~\eqref{eq:lambdaEps} (as it can be readily resolved by canceling some terms). Therefore, we limit our search to the $\lambda$ described by Equation~\eqref{eq:lambdaEps} that serves as the root of $(a+(m-1)b)$. More specifically, we deal with
\begin{equation*}
\epsilon((\Lambda_\epsilon)_{ww} + (m-1)(\Lambda_\epsilon)_{ww'})) = A(\gamma)\epsilon^2 + B(\gamma)\epsilon + C(\gamma) + O(\epsilon^3).
\end{equation*}
Then, we can proceed in the same way as before, deriving $\gamma$ such that $B(\gamma) = 0$ and $\epsilon^\pm$ from the final quadratic expression
\begin{align*}
    \gamma &= \displaystyle \frac{- 3 k^{2} m + 4 k^{2} v + \lambda_{2} m - m r}{4 k v \left(k^{2} + \lambda_{2} - r\right)},\\
    \epsilon^\pm &= \pm \frac{\sqrt{k}\sqrt{m}}{\sqrt{2v}\sqrt{k+1}} + O\left(\frac{1}{v}\right).
\end{align*}

The expressions for $t_\text{opt}$ and $p_\text{succ}$ still derive from the formulas provided by~\cite{LPST23}. As the calculations in the cited paper primarily focus on the two marked cases, and we have already used its final formulas in the previous subsection, we will now elaborate on these expressions in more detail. We start by observing that~\cite{LPST23} provides us with an expression for the probability of locating a marked vertex as a function of time, following certain assumptions (which we have also considered). The expression is as follows
\begin{equation}
    p(t) = 4 |\braket{\lambda|\psi(0)}|^2 \sum_{w\in W} |\braket{w|\lambda}|^2\sin^2 \epsilon t+o(1) + o(\epsilon t),
\end{equation}
which implies that $t_\text{opt} = \frac{\pi}{2\epsilon}$, giving us 
\begin{equation}
    t_\text{opt} = \frac{\pi\sqrt{k+1}}{\sqrt{2k}\sqrt{m}}\sqrt{v} + O\left(\frac{1}{\sqrt{v}}\right). \label{eq:toptMMSS}
\end{equation}
The success probability is
\begin{equation}
    p_\text{succ} = 4 |\braket{\lambda|\psi(0)}|^2 \sum_{w\in W} |\braket{w|\lambda}|^2. \label{eq:pSucc}
\end{equation}
First, we need to find each $|\braket{w|\lambda}|$, taking note that the definition of $\Lambda_\lambda$ in~\cite{LPST23} establishes that the vector with $|\braket{w|\lambda}|$ as its coordinates is a 0-eigenvector. Based on our choice of $\epsilon$, we have $u = (1,1,...,1)$ as a 0-eigenvector. This is a multiple of the desired vector, and it already demonstrates that $|\braket{w|\lambda}| = |\braket{w'|\lambda}|$ for every pair $(w,w')$. We find a correction factor $c$, which is provided in~\cite{LPST23} and it follows that:
\begin{equation*}
    \braket{w|\lambda} = c = \frac{k}{\sqrt{2}\sqrt{m}\sqrt{k^2 - \lambda_2 + r}} + \left(\frac{1}{\sqrt{v}}\right).
\end{equation*}
Knowing this term, we can also calculate
\begin{align*}
    \braket{\psi(0)|\lambda} &= - \frac{1}{\lambda + \gamma\phi_0}\sum_{w\in W} \braket{w|\lambda}\braket{\psi_0|w}\\
    &=-\frac{\sqrt{2}}{2} + o(1).
\end{align*}
Substituting all the calculated terms in Equation~\eqref{eq:pSucc} finally yields
\begin{equation}
    p_\text{succ} = \frac{k}{k+1}+ O\left(\frac{1}{\sqrt{v}}\right).
\end{equation}

\subsubsection*{More general cases}
Another potential generalization is the scenario where the subgraph induced by the marked vertices is regular with parts of equal size. Although calculating the determinant remains a challenging task in this case, we can leverage the fact that the sum of all rows is the same. Consequently, we can work with the eigenvector consisting entirely of ones and find the $\lambda$ value that results in a zero eigenvalue.

Let $m$ represent the number of marked vertices in each part of the bipartite graph, and let $d$ denote the number of marked vertices each marked vertex is connected to. Additionally, let us define $a$ as $(\Lambda_\lambda)_{ww'}$ when $w$ and $w'$ are adjacent, $b$ as the expression when they are in different parts and non-adjacent, and $c$ as the expression for vertices in the same part. Then, the sum of each row in $\Lambda_\lambda$ will be equal to
\begin{equation*}
    (\Lambda_\lambda)_{ww} + d a + (m-d) b + (m-1) c.
\end{equation*}
It is the eigenvalue associated with vector $(1,1,...,1)$. We wish to reduce this eigenvalue to zero. We employ the same strategy of multiplying the expression by $\epsilon$, identifying the $\gamma$ that cancels out the linear term, and then calculating the values of $\epsilon^\pm$. We obtain
\begin{align*}
    \gamma&=\frac{k\left(km - dv + (m-v)\sqrt{k-\lambda_2}\right)}{v\sqrt{k-\lambda_2}\,(k-\lambda_2 - k^2)},\\
    \epsilon^\pm &= \pm  \frac{\sqrt{m(k+d\sqrt{k})}}{\sqrt{k + 1}} \frac{1}{\sqrt{v}} + O\left(\frac{1}{v}\right).
\end{align*}
From this, we can directly derive
\begin{equation}
t_\text{opt} = \displaystyle \frac{\pi\sqrt{k+1}}{2\sqrt{m(k+d\sqrt{k})}}\sqrt{v} + O\left(\frac{1}{v}\right).\label{eq:toptMD}
\end{equation}
In the preceding equation, it's worth noting that we can revert to the expression for two marked vertices in separate parts ($m = 1$) in both the adjacent ($d=1$) and non-adjacent ($d=0$) cases. The scenario where the two marked vertices are within the same part constitutes a subcase of the discussion from the prior subsection.

To calculate the success probability, we follow the method used in the previous subsection calculating $\braket{\lambda|w}$ and $\braket{\psi(0)|\lambda}$ to obtain
\[p_\text{succ} = \frac { 4\,\sqrt {k} \left( k+1
 \right)\left( d+\sqrt {k} \right)^{3} }{ \left( (k+1)(d+2\sqrt{k}) + 2d\sqrt{k}
 \right)^{2}} + O\left(\frac{1}{v}\right).
\]

In every case we analyzed, it's evident that the primary factor affecting the success probability is independent of both the number of vertices and the number of marked vertices. In every instance, when we enhance the degree $k$ of the graph and the number of nodes $v$, the probability approaches 1. It's also crucial to highlight that in the final scenario, where the marked vertices are connected, the success probability rises in proportion to the degree of connectivity among the marked vertices. Here, $d$ represents the degree of the induced subgraph.

Given that the success probability is $o(1)$, our main concern becomes determining the optimal number of steps to infer the time complexity of the search algorithm. While the optimal time is $O(\sqrt{v})$, Equations~\eqref{eq:toptMMSS} and~\eqref{eq:toptMD} elucidate how the number of marked vertices $m$ and their connectivity (expressed by $d$) can reduce the algorithm's runtime.

A notable point concerning $t_\text{opt}$ in Equation~\eqref{eq:toptMD}, particularly in cases where each part of the bipartite graph features $m$ marked vertices, emerges as follows: When $k > m$, we have the option to choose a set of marked vertices inducing a complete bipartite subgraph, where $d = m$. This implies that each marked vertex connects to the marked vertices in the opposing part. Under these conditions, $t_\text{opt} = O\left(\frac{\sqrt{v}}{m}\right)$, assuming $m$ stays below a certain fixed $k$. Such behavior is not common in quantum searches that involve multiple marked vertices.

\section{Final remarks}\label{sec:conc}

In this work, we investigated quantum search algorithms on bipartite graphs using continuous-time quantum walks, with a particular focus on the significance of combinatorial $t$-designs. We determined the eigenvalues and eigenvectors of symmetric $t$-designs. By carefully analyzing these designs and their incidence matrices, we identified a subset with symmetries that are particularly useful for determining the time complexity in scenarios with multiple marked vertices. Our findings highlight the effectiveness of the continuous-time quantum walk on certain symmetric $t$-designs. We show that these designs can achieve a running time of $O(\sqrt{n})$, where $n$ is the number of vertices in the corresponding bipartite graph, regardless of the number of marked vertices. Additionally, the success probability remains consistently $o(1)$, but it converges to unity in certain cases. These results are among the first to study multiple marked elements on complex geometric structures where the time complexity of spatial search algorithms can be analytically derived using the continuous-time model.

This work deepens our understanding of the relationship between quantum walks and combinatorial graph theory, especially within the context of symmetrical $t$-designs. Additionally, it aids in paving the path for analytically determining the time complexity of quantum-walk-based search algorithms on bipartite graphs with multiple marked vertices, a notably challenging task in the continuous-time scenario. Anticipated future research offers the prospect of more comprehensive results, potentially broadening the relevance of non-symmetric $t$-designs in quantum search algorithms.

\section*{Conflict of interest}
The authors declare that there is no conflict of interest.

\section*{Data availability}
All data generated or analyzed during this study are included in this published article.

\section*{Acknowledgements}
We are indebted to Chris Godsil and Qiuting Chen for their insightful discussions, which significantly contributed to the formulation of Theorem 1~\cite{phdthesis}.
The work of P. Lugão was supported by CNPq grant number 140897/2020-8, and FAPERJ grant number E-26/202.351/2022. The work of R. Portugal was supported by FAPERJ grant number CNE E-26/200.954/2022, and CNPq grant numbers 308923/2019-7 and 409552/2022-4. The authors have no competing interests to declare that are relevant to the content of this article.

\bibliographystyle{unsrt}


\end{document}